\documentclass{smjour}
\usepackage{graphth}
\usepackage{amsmath,amssymb,amsfonts,latexsym}
\usepackage{color}
\usepackage{graphics,graphicx}
\usepackage{algorithm,algorithmic}
\usepackage{slashed}
\usepackage{relsize}
\DeclareMathOperator*{\BOX}{\Box}
\newcommand{\eqcl}{\mathrel{\mathsmaller{\mathsmaller{^{\boldsymbol{\sqsubseteq}}}}}}
\newcommand{\DELTA}{\mathrel{\delta}}

\newtheorem{obs}{Observation}

\newcommand*{\N}{N_{\varphi}}
\newcommand*{\Npsi}{N_{\varpsi}}

\newcommand*{\No}{N_{\overline\varphi}}

\newcommand*{\G}{G_{\varphi}}

\newcommand*{\Gopsi}{G_{\overline\varpsi}}

\newcommand*{\Gi}{G_{\varphi_i}}

\newcommand*{\Gkp}{G_{\varphi_{k+1}}}
\newcommand*{\Go}{G_{\overline\varphi}}

\newcommand*{\Goi}{G_{\overline\varphi_i}}
\newcommand*{\Goj}{G_{\overline\varphi_j}}

\newcommand*{\Gokp}{G_{\overline\varphi_{k+1}}}
\newcommand*{\Goe}{G_{\overline\varphi_1}}
\newcommand*{\Gon}{G_{\overline\varphi_n}}

\newcommand*{\vp}{\varphi}

\newcommand*{\vpo}{\overline\varphi}
\newcommand*{\mc}{\mathcal}
\newcommand*{\inters}{V_R}

\newcommand{\varpsi}{\psi}

\def\jfootline{\iftitle\global\titlefalse
\vtop to0pt{\everypar={}\leftskip=0pt\rightskip=0pt
\vskip12pt
\iffootlinehrule%
\hrule height .3pt
\vskip8pt\else\vskip1sp\fi
\parindent=0pt \parskip=0pt \parfillskip=0pt 
\baselineskip=9pt \footlinefont
\phantom{\currstyle\ 
Vol.\ \thevolume, \thepages\ (\theyear)}\hfill\ \vskip1pt
\hbox to\textwidth{\phantom{\smcopyright\hskip3pt\theyear\ 
John Wiley \& Sons, Inc.\hfill CCC \thecccline}}
\vskip1pt
\iffootlinefolio
\vskip4pt
\hbox to\textwidth{\hfill\foliofont\the\c@page}\fi
\vss}\else\hfill\fi}
\def\abstract#1{{\leftskip=\abstractmargin \rightskip=\leftskip
\vskip\aboveabstractskip
\abstractsize\ifabstractname
\abstractnamefont\noindent 
ABSTRACT\vskip\belowabstractnameskip\fi\abstractfont
\noindent#1 {\footnotesize%
\phantom{\bf\smcopyright\hskip3pt\theyear\ John Wiley \& Sons,~Inc.}}
\vskip\belowabstractskip}}

\makeatletter
\def\jheadline{\hbox to\textwidth{\iftitle\else\ifodd\c@page
{\hfill\headlinesize\headtextfont\thetitle}%
\enspace{\foliofont\the\c@page}
\else{\foliofont\the\c@page}
\enspace{\headlinesize\headtextfont\thetitle}\hfill\fi\fi}}
\makeatother

\begin{document}

\sloppy


\yearofpublication{(2012)}
\volume{ }
\issuenumber{ }
\cccline{ }

\authorrunninghead{Hellmuth, Ostermeier, Stadler}
\titlerunninghead{Square Property, Equitable Partitions, and Product-like Graphs}

\setcounter{page}{1} 



\title{Square Property, Equitable Partitions, and Product-like Graphs}

\author{Marc Hellmuth}
\affil{%
Center for Bioinformatics, 
Saarland University, Building E 2.1, 
D-66041 Saarbr{\"u}cken, Germany\\
{\scriptsize e-mail: marc.hellmuth@bioinf.uni-sb.de}}

\author{Lydia Ostermeier}
\affil{%
Max-Planck-Institute for Mathematics in the Sciences,
Inselstrasse 22, D-04103 Leipzig,  Germany \\[0.2cm]
Interdisciplinary Center for Bioinformatics
University of Leipzig,
H{\"a}rtelstrasse 16-18, D-04107 Leipzig, Germany\\
{\scriptsize e-mail: glydia@bioinf.uni-leipzig.de}}

\authors{Peter F.\ Stadler}
\affil{Bioinformatics Group, Department of Computer Science and
Interdisciplinary Center for Bioinformatics
University of Leipzig,
H{\"a}rtelstrasse 16-18, D-04107 Leipzig, Germany\\[0.2cm]
Max Planck Institute for Mathematics in the Sciences
Inselstrasse 22, D-04103 Leipzig, Germany\\[0.2cm]
RNomics Group,
  Fraunhofer Institut f{\"u}r Zelltherapie und Immunologie,
  Deutscher Platz 5e, D-04103 Leipzig, Germany\\[0.2cm]
Department of Theoretical Chemistry, University of Vienna,
  W{\"a}hringerstrasse 17, A-1090 Wien, Austria\\[0.2cm]
Santa Fe Institute, 1399 Hyde Park Rd., Santa Fe, NM87501,
  USA\\
{\scriptsize e-mail: studla@bioinf.uni-leipzig.de}}

\newpage

\abstract{Equivalence relations on the edge set of a graph $G$ that satisfy
  restrictive conditions on chordless squares play a crucial role in the
  theory of Cartesian graph products and graph bundles. We show here that
  such relations in a natural way induce equitable partitions on the vertex
  set of $G$, which in turn give rise to quotient graphs that can have a
  rich product structure even if $G$ itself is prime.}

\begin{article} 

\section{Introduction}
\label{sec:Preliminaries}

Sabidussi \cite{Sabidussi:60} and later Vizing \cite{Viz:63} showed that
every finite connected graph has a unique prime factorization w.r.t.\ the
Cartesian product. This Cartesian product structure is naturally understood
in terms of an equivalence relation $\sigma$ on the edge set $E(G)$ that
identifies the fibers as the connected components of the subgraphs of $G$
that are induced by a single equivalence class of $\sigma$
\cite{Sabidussi:60}. The first polynomial time algorithm to compute the
factorization of an input graph \cite{Feigenbaum85:CartProd} explicitly
constructs $\sigma$ starting from another, finer, relation $\delta$. The
product relation $\sigma$ was later shown to be simply the convex hull
$\mathfrak{C}(\delta)$ of the relation $\delta$ \cite{Imrich:94}.

Graph bundles \cite{Pisanski:83}, the combinatorial analog of the
topological notion of a fiber bundle \cite{Husemoller:93}, are a common
generalization of both Cartesian products \cite{Hammack:11a} and covering
graphs \cite{Abello:91}. A slight modification of the relation $\delta$
turns out to play a fundamental role for the characterization of graph
bundles \cite{Zmazek:02} and forms the basis of efficient algorithms to
recognize Cartesian graph bundles \cite{Imrich:97,Zmazek:02a,Zmazek:02}.
Here we introduce a further generalization, termed USP-relations, that
still retains the salient properties of $\delta$.

The connected components of a given equivalence class of the product
relation $\sigma$, i.e., the fibers of $G$ w.r.t.\ to a given factor $F$,
form a natural partition $\mathcal{P}_F$ of the vertex set of $G$. It is
well known (see e.g.\ \cite{Hammack:11a}) that $G$ then has a
representation as $G \cong (G/\mathcal{P}_F) \square F$. It is of interest,
therefore, to consider quotient graphs of Cartesian products in a more
systematic way.

Equitable partitions of graphs \cite{FS:11,FPT:08} were originally introduced as a means of
simplifying the computation of graph spectra \cite{Schwenk:74} and walks on
graphs \cite{Godsil:80}.  A series of recent results on so-called perfect
state transfer revealed a close
connection between equitable partitions of the vertex set of $G$, the
corresponding quotient graphs, and the Cartesian product structure of $G$
\cite{BFF+12,YGB+11}.

Here, we show that equitable partitions on the vertex set $V(G)$ are
induced in a natural way in a more general setting, namely by equivalence
relations that are coarsenings of relations with the unique square property
on the edge set $E(G)$. The quotient graphs w.r.t.\ these equitable
partitions exhibit a natural, rich product structure even when $G$ itself
is prime. It can therefore be regarded as an ``approximate graph
  product'', albeit in a somewhat different sense than the deviations from
  product structures explored e.g.\ in \cite{HIKS-08,HIKS-09,Hel:11}.

\section{Background and Preliminaries}
\label{sec:prelim}

\subsection{Basic Definitions and Notation}

In the following we assume that $G$ is a finite connected graph with vertex
set $V=V(G)$ and edge set $E=E(G)$. A graph $H$ is a subgraph of $G$ if
$V(H)\subseteq V(G)$ and $E(H)\subseteq E(G)$. $H$ is an \emph{induced
  subgraph} of $G$ if $x,y\in V(H)$ and $(x,y)\in E(G)$ implies $(x,y)\in
E(H)$.  An induced cycle on four vertices is called \emph{chordless
  square}.

\bigskip\noindent\textbf{Relations.}
We will consider equivalence relations $R$ on $E$, i.e., $R\subseteq
E\times E$ such that (i) $(e,e)\in R$, (ii) $(e,f)\in R$ implies $(f,e)\in
R$ and (iii) $(e,f)\in R$ and $(f,g)\in R$ implies $(e,g)\in R$. The
equivalence classes of $R$ will be denoted by Greek letters,
$\varphi\subseteq E$. We will furthermore write $\varphi \eqcl R$ for 
mean that $\varphi$ is an equivalence class of $R$. 

A relation $Q$ is finer than a relation $R$ while the relation $R$ is
coarser than $Q$ if $(e,f)\in Q$ implies $(e,f)\in R$, i.e, $Q\subseteq
R$. In other words, for each class $\vartheta$ of $R$ there is a collection
$\{ \chi | \chi\subseteq \vartheta\}$ of $Q$-classes, whose union equals
$\vartheta$. Equivalently, for all $\varphi\eqcl Q$ and $\psi\eqcl R$ we
have either $\varphi\subseteq \psi$ or $\varphi\cap\psi=\emptyset$.

To make this paper easier to read we denote in the following refinements of
a given relation $R$ by $Q$ and coarse grainings of $R$ by $S$, so that
$Q\subseteq R \subseteq S$.

For a given equivalence class $\varphi \eqcl R$ and a vertex $u\in V(G)$ we
denote the set of neighbors of $u$ that are incident to $u$ via an edge in
$\varphi$ by $N_{\varphi}(u)$, i.e.,
\begin{equation*}
N_{\varphi}(u):= \{v\in V(G)\mid [u,v]\in\varphi\} \,.
\end{equation*}

\bigskip\noindent\textbf{Equitable Partitions.}  A partition $\mathcal{P}$
of the vertex set $V(G)$ of a graph $G$ is \emph{equitable} if, for all
(not necessarily distinct) classes $A,B\in\mathcal{P}$ every vertex $x\in
A$ has the same number
$$ m_{AB} :=  |N_G(x) \cap B| $$
of neighbors in $B$.  The matrix $\mathbf{M}=\{m_{AB}\}$ is called
\emph{partition degree matrix}.  

\bigskip\noindent\textbf{Quotient Graphs.} Let $G$ be a graph and
$\mathcal{P}$ be a partition of $V(G)$. The (undirected) \emph{quotient
  graph} $G/\mathcal{P}$ has as its vertex set $\mathcal{P}$, i.e., the
classes of the partition. There is an edge $[A,B]$ for $A,B\in\mathcal{P}$
if and only if there are vertices $a\in A$ and $b\in B$ such that $[a,b]\in
E(G)$. Note that there is a loop $[A,A]$ unless the class $A$ of
$\mathcal{P}$ is an independent set.  

\bigskip\noindent\textbf{Weighted Quotient Graphs.} 
Let $G$ be a graph and let $\mathcal{P}$ be an equitable partition of
$V(G)$ with partition degree matrix $\mathbf{M}$. The \emph{directed
  weighted quotient graph} $\overrightarrow{G/\mathcal{P}}$ has vertex set
$V(\overrightarrow{G/\mathcal{P}}) = \mathcal{P}$ and directed edges
$(A,B)$ from $A$ to $B$ with weight $m_{AB}$ iff $m_{AB}\ge1$. Note that
$\overrightarrow{G/\mathcal{P}}$ has loops whenever $m_{AA}\geq 1$.

By construction, $m_{AB}\geq 1$ implies $m_{BA}\geq 1$. Hence
$\overrightarrow{G/\mathcal{P}}$ has a well-defined underlying undirected and
unweighted graph, which obviously coincides with $G/\mathcal{P}$. The
underlying simple graph, obtained by also omitting the loops, will be
denoted by $\mathcal{N}(\overrightarrow{G/\mathcal{P}})=
\mathcal{N}(G/\mathcal{P})$.

\bigskip\noindent\textbf{Cartesian Graph Product.} 
The \emph{Cartesian product} $G\BOX H$ has vertex set $V(G\BOX H) =
V(G)\times V(H)$; two vertices $(g_1,h_1)$, $(g_2,h_2)$ are adjacent in
$G\BOX H$ if $(g_1,g_2)\in E(G)$ and $h_1=h_2$, or $(h_1,h_2)\in E(G_2)$
and $g_1 = g_2$.

Cartesian products generalize in a natural way to directed edge-weighted
graphs (with loops allowed). Their Cartesian product $G\BOX H$ has the 
edge weights 
\begin{equation}
  \label{eq:weights}
  m((g_1,h_1),(g_2,h_2))= 
  \begin{cases}
    m_G(g_1,g_2),\qquad &\text{iff } h_1=h_2\text{ and }g_1\neq g_2\\
    m_H(h_1,h_2),\qquad &\text{iff } g_1=g_2\text{ and }h_1\neq h_2\\
    m_G(g_1,g_2)+m_H(h_1,h_2),\quad &\text{iff } g_1=g_2\text{ and }h_1=h_2\\
    0,\qquad&\text{otherwise}
  \end{cases}
\end{equation}
where $m_G(g_1,g_2)$ and $m_H(h_1,h_2)$ denotes the edge weight of the arcs
$(g_1,g_2)$ in $G$ and $(h_1,h_2)$ in $H$, resp., The absence of such an
arc is equivalent to $m_X(x_1,x_2)=0$ for $X\in\{G,H\}$. The Cartesian
product of the underlying undirected and unweighted graphs is obtained by ignoring the
directions and weights in the product graph.

The Cartesian product is associative and commutative. Every finite
connected graph $G$ has a decomposition $G=\BOX_{i=1}^n G_i$ into prime
factors that is unique up to isomorphism and the order of the
factors \cite{Sabidussi:60}.

The mapping $p_i: V(\Box_{i=1}^n G_i)\rightarrow V(G_i )$ defined by
$p_i(v) = v_i$ for $v = (v_1,v_2,\ldots,v_n)$ is called \emph{projection}
on the $i$-th factor of $G$. The induced subgraph $G_i^w$ of $G$ with
vertex set $V(G_i^w)=\{v\in V(G)\mid p_j(v)=w_j, \text{ for all } j\neq
i\}$ is called \emph{$G_i$-layer through $w$}. It is isomorphic to $G_i$.

An equivalence relation $R$ on the edge set $E(G)$ of a Cartesian
product $G=\Box_{i=1}^n G_i$ of (not necessarily prime) graphs $G_i$
is a \emph{product relation} if $e \mathrel{R} f$ holds  
if and only if there exists a $j\in \{1,\ldots,n\}$ such that 
$|p_j(e)|=|p_j(f)|=2$.

\bigskip\noindent\textbf{Cartesian Graph Bundles.}  Given two graphs $G$
and $H$, a map $p:G\to H$ is called a \emph{graph map} if $p$ maps adjacent
vertices of $G$ to adjacent or identical vertices in $B$ and edges of $G$
to edges or vertices of $B$. Graph maps are also known as weak
homomorphisms \cite{Hammack:11a}. For instance, the projections $p_i$ of
product graphs to their factors are graph maps.

A graph $G$ is a \emph{Cartesian graph bundle} if there are two graphs $F$,
the fiber, and $B$ the base graph, and a graph map $p:G\to B$ such that:
For each vertex $v\in V(B)$, $p^{-1}(v)\cong F$ and for each edge $e\in
E(B)$ we have $p^{-1}(e)\cong K_2\square F$. The triple $(G,p,B)$ is called
a presentation of $G$ as a Cartesian graph bundle. If $G=\Box_{i=1}^n G_i$
is a product, then $(G,p_j,G_j)$ is a bundle presentation of $G$ with fiber
$\Box_{i=1, i\neq j}^n G_i$ for all $1\leq j\leq n$.

\subsection{From Edge Partitions to Vertex Partitions}

We start from an equivalence relation $R$ on $E(G)$. Let $\varphi\eqcl R$.
An edge $e$ is called $\vp$-edge if $e\in \vp$.  The subgraph $G_{\varphi}$
has vertex set $V(G)$ and edge set $\varphi$. The connected components of
$G_{\varphi}$ containing vertex $x\in V(G)$ are denoted by $G_{\varphi}^x$.

By construction, the set
\begin{equation*}
\mathcal{P}^R_{\varphi}:=\left\{V(G^x_{\varphi})\mid x\in V(G)\right\}
\end{equation*} 
is a partition of $V(G)$ for every $\vp\eqcl R$. The quotient graph 
$G/\mathcal{P}^R_{\varphi}$ has as its vertex sets the connected
components $G^x_{\varphi}$ and edges $(G^x_{\varphi},G^y_{\varphi})$
if and only if there are $x'\in V(G^x_{\varphi})$ and $y'\in V(G^y_{\varphi})$
with $(x',y')\in E(G)$.

The projection $p_{\varphi}:G\to G/P_{\varphi}^R$ defined by $x \mapsto
G^x_{\varphi}$ is a graph map. If $(x,y)\in \varphi$ then $y\in
V(G^x_{\varphi})$ and hence $G^x_{\varphi}=G^y_{\varphi}$. Thus, we have a
loop in the quotient graph $G/\mathcal{P}^R_{\varphi}$ for every
$V(G^x_{\varphi})\ne\{x\}$. Edges that do not form a loop 
in $G/P_{\varphi}^R$ thus arise only from $(x,y)\in E\setminus\varphi$.

In the following we will be interested in particular in the complements of
$R$-classes, i.e., in $\overline{\varphi} := E\setminus \varphi$. The
corresponding subgraphs are denoted by $G_{\overline{\varphi}}$, with
connected components $G_{\vpo}^x$ for a given $x\in V(G)$.  For later
reference we note following simple
\begin{obs}
\label{obs:not-vp-path}
It holds $y\in V(\Go^x)$ if and only if there is a path
$P:=(x=x_0,x_1,\ldots x_k=y)$ from $x$ to $y$ such that
$[x_i,x_{i+1}]\notin \varphi$ for all $0\leq i\leq k-1$.
\end{obs}

Just like $\mathcal{P}^R_{\varphi}$, the set
\begin{equation}\label{eq:Povp}
\mathcal{P}^R_{\vpo}:=\left\{V(\Go^x)\mid x\in V(G)\right\}
\end{equation} 
is a partition of $V(G)$ for every $\vp\eqcl R$. To see this, we note that 
$x\in V(\Go^x)$ holds for all $x\in V(G)$. Thus, $P\neq\emptyset$ for all 
$P\in \mathcal P^R_{\vpo}$ and $\bigcup_{P\in \mathcal P^R_{\vpo}} P =
V(G)$. Furthermore, $V(\Go^x) \cap V(\Go^y) \neq \emptyset$ if
and only if $x$ and $y$ are in same connected component w.r.t.\ $\vpo$, i.e.,
if and only if $V(\Go^x) = V(\Go^y)$. 
Note, Graham and Winkler showed in \cite{GrahamWinkler:84} that 
the particular defined equivalence relation $R=\theta^*$ on $E(G)$, the so-called 
Djokovi\'{c}-Winkler relation, induces a canonical isometric embedding of 
a graph $G$ into a  Cartesian product 
$\Box_{\varphi \eqcl R} G_{\varphi} / \mc P_{\vpo}^{R}$.
Moreover, Feder showed that 
if we choose $R=(\theta \cup \tau)^*$ then 
$G \cong \Box_{\varphi \eqcl R} G_{\varphi} / \mc P_{\vpo}^{R}$
and thus,  $R$ is the product relation $\sigma$, see
\cite{Feder:92}.

We furthermore will need the intersections
\begin{equation*}
\inters(x) := \bigcap_{\vp \eqcl R } V( G_{\vpo}^x)\,.
\end{equation*}
These sets form the classes of the common refinement 
of the partitions $\mathcal{P}^R_{\vpo}$, $\varphi\eqcl R$, i.e.,
\begin{equation}\label{eq:mcPR}
\mathcal{P}^R :=\left\{\bigcap_{\vp \eqcl R } 
  V(G_{\overline\varphi}(x))\mid x\in V(G)\right\}
= \left\{\inters(x)\mid x\in V(G)\right\}
\end{equation}
is again a partition of $V(G)$. 

\begin{lemma}
Let $Q$ and $R$ be two equivalence relations on $E(G)$ so that 
$Q$ is finer than $R$. Then $V_R(x)\subseteq V_Q(x)$.
\end{lemma}
\begin{proof}
Consider two equivalence classes $\varphi,\varpsi\eqcl Q$. 
From $\overline{\varphi\cup\varpsi}=\overline{\varphi}\cap\overline{\varpsi}$ 
we observe that $G_{\overline{\varphi\cup\varpsi}}$ is a subgraph of both
$G_{\overline{\varphi}}$ and $G_{\overline{\varpsi}}$. This remains true
for the connected components containing a given vertex $x\in V$, and hence
\begin{equation*}
  V(G^x_{\overline{\varphi\cup\varpsi}})\subseteq
  V(G^x_{\overline{\varphi}}) \cap V(G^x_{\overline{\varpsi}})\,.
\end{equation*}
Using this observation we compute  
\begin{equation*}
V_R(x) =
\bigcap_{\vartheta\eqcl R} V\left( G^x_{\overline{\vartheta}} \right)
=
\bigcap_{\vartheta\eqcl R} V\left(
   G^x_{\overline{\bigcup_{\chi\subseteq\vartheta}\chi}} \right)
   \subseteq
\bigcap_{\vartheta\eqcl R} \bigcap_{\chi\subseteq\vartheta}
    V\left( G^x_{\overline{\chi}} \right) =
\bigcap_{\chi\eqcl Q} V\left( G^x_{\overline{\chi}} \right) =
V_Q(x)
\end{equation*}
\end{proof}

\par\noindent
Thus, a coarser equivalence relation $R$ on $E(G)$ leads to smaller sets
$V_R(x)$, and hence to a finer partition $\mathcal{P}^R$ of the vertex set.

\subsection{The Square Property and the Unique Square Property}

\begin{definition}
Two edges $e,f\in E(G)$ are in the \emph{relation $\delta$}, $e\DELTA f$, 
if one of the following conditions is satisfied:
\begin{itemize}
\item[(i)]   $e$ and $f$ are opposite edges of a chordless square.
\item[(ii)]  $e$ and $f$ are adjacent and there is no chordless square 
             containing $e$ and $f$. 
\item[(iii)] $e=f$.
\end{itemize}   
\end{definition}

\begin{definition}
  An equivalence relation $R$ on $E(G)$ has the 
  \emph{unique square property} if it satisfies
  \begin{itemize}
  \item[(S1)] Any two adjacent edges $e$ and $f$ from distinct equivalence
  classes span a unique chordless square with opposite edges in the same
  equivalence class $R$.
  \end{itemize}
  $R$ has the \emph{square property} if it satisfies in addition 
  \begin{itemize}
  \item[(S2)] The opposite edges of any chordless square belong to the 
    same equivalence class.
  \end{itemize}
\end{definition}

The unique square property was introduced by Zmazek and 
{\v{Z}}erovnik \cite{Zmazek:02} as a feature of the so-called 
fundamental factorizations of graph bundles over simple 
bases. The results derived below therefore hold in particular 
also for this type of graph bundles.

Relations with the unique square property do not need to satisfy the
square property as shown by the counterexample in Figure
\ref{fig:CounterUSP}. On the other hand, from the definition it is clear,
that every equivalence relation $R$ on $E(G)$ that has the 
square property also has the unique square property.
It
has been noted, e.g.\ in \cite{Zmazek:02}, that $\DELTA$ has the unique
square property.

\begin{figure}[tbp]
  \centering
  \includegraphics[bb= 140 412 304 501, scale=0.9]{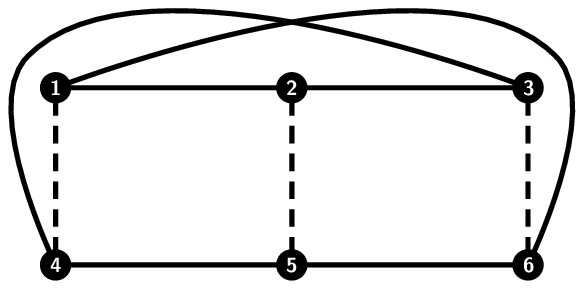}
  \caption{The square property and the unique square property are not 
    equivalent. The line styles distinguish the two classes of the
    equivalence relation $R$ on the edges. It has the unique square
    property (S1).  The edges $[1,2]$ and $[1,4]$ span two
    chordless squares $Sq_1=[1,2,5,4]$ and $Sq_2=[1,2,3,4]$ of which  
    $Sq_1$ has opposite edges in the same equivalence class. The square 
    $Sq_2$ thus violates (S2).}
 \label{fig:CounterUSP}
\end{figure}

The following observation has been used implicitly e.g.\ in 
\cite{Imrich:97,Zmazek:02}.
\begin{proposition}
An equivalence relation $R$ on $E(G)$ has the square property if and only if  
$\delta\subseteq R$. 
\end{proposition}
\begin{proof}
  Let $R$ be an equivalence relation on $E(G)$ and $\delta\subseteq
  R$. Then Condition $(i)$ in the definition of $\DELTA$ directly
  implies Condition (S2).
  Let $e,f$ be two adjacent edges and suppose $(e,f)\notin R$.
  Then there must exist a square containing both edges, otherwise,
  by condition $(ii)$, $(e,f)\in\delta\subseteq R$,  a contradiction.
  Let this square consist of edges $e,f,e',f'$
  such that $e'$ is opposite edge to $e$ and $f'$ is opposite edge to $f$.
  Then condition $(i)$ implies $(e,e'),(f,f')\in\delta\subseteq R$.
  Assume $e,f$ are contained in another square consisting of edges
  $e,f,e'',f''$ such that
  $e''$ is opposite edge to $e$ and $f''$ is opposite edge to $f$.
  Then there is also a square consisting of edges $e',f',f'',e''$
  such that $e''$ is opposite edge to $f'$ and $f''$ is opposite edge to $e'$.
  Again condition $(i)$ implies $(e,e''),(f,f'')\in\delta\subseteq R$ as well as
  $(e'',f'),(f'',e')\in\delta\subseteq R$.
  Finally, by transitivity it follows $(e,f)\in R$, a contradiction.
  Thus, $R$ has the square property.

  Let $R$ be an equivalence relation on $E(G)$ with the square property. We
  have to show that $e\DELTA f$ implies $e\mathrel{R}f$ for all edges
  $e,f$. First assume $e\DELTA f$ such that $e$ and $f$ are not
  adjacent. Hence, either Condition (i) or (iii) is fulfilled which
  immediately implies $e\mathrel{R}f$. Now, let $e$ and $f$ be adjacent and
  assume for contraposition that $e\not\mathrel{R}f$.  Thus, by condition
  $(S1)$ there is a chordless square spanned by $e$ and $f$ and therefore
  $e$ and $f$ do not satisfy condition $(ii)$. Hence, $e\not\DELTA f$ which
  completes the proof.
\end{proof}

The transitive closure $\delta^*$ of $\delta$ is therefore the
finest equivalence relation on $E(G)$ that has the square
property. Furthermore, an equivalence relation $R$ has the square property
if and only if its classes are unions of equivalence classes of
$\delta^*$. Therefore, if $R$ has the square property and $R\subseteq S$,
then the coarser equivalence relation $S$ also has the square property.

In contrast, there is no finest equivalence relation that has the unique
square property. Moreover, if an equivalence relation $R$ satisfies the
unique square property, it is still possible that there exists
a coarser equivalence relation $S\supset R$ that does not have the
unique square property, as shown by the example in
Figure~\ref{fig:CounterUnion}.

\begin{figure}[tbp]
  \centering
  \includegraphics[bb= 250 590 400 730, scale=0.75]{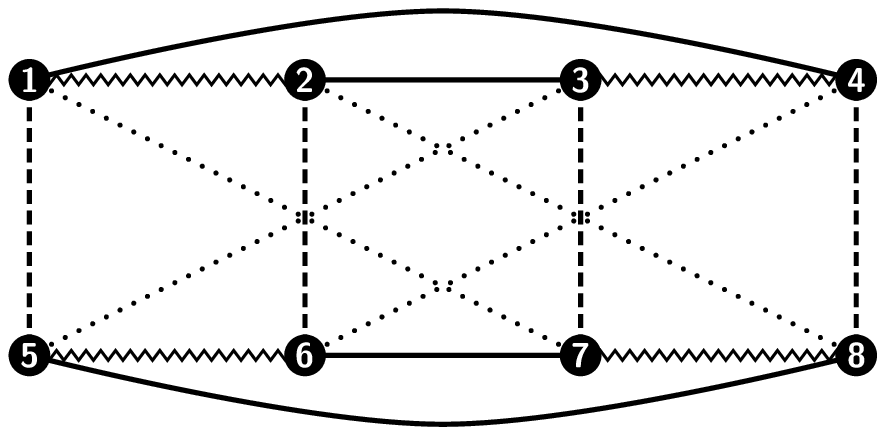}
  \caption{The equivalence relation $Q$ on the edge set $E(G)$ of the
    ``diagonalized  cube'' $G$ has the four equivalence classes
    $\vp_1,\vp_2,\vp_3$ and $\vp_4$ depicted by solid, zigzag, dotted and
    dashed edges, respectively.  One easily checks that $Q$ has the
    unique square property. The relation $R$ with classes
    $\psi_1=\vp_1\cup\vp_2$ and $\psi_2=\vp_3\cup\vp_4$, however, does
    not have the unique square property, because the edges $[1,5]$ and
    $[1,2]$ span \emph{two} squares $(1,5,6,2)$ and $(1,5,6,4)$ with
    opposite edges belonging to the same class. 
    Clearly, $R$ is a USP-relation}
  \label{fig:CounterUnion}
\end{figure}

This observation motivates us to consider a slightly more general 
set of equivalence relations. 

\begin{definition}
  An equivalence relation $R$ on the edge set of a connected graph $G$ 
  is called a \emph{USP-relation} if there exists a finer equivalence 
  relation $Q\subseteq R$ that satisfies the unique square property.
\end{definition}

\begin{obs} 
\label{obs:squares} 
  If $R$ is a USP-relation on the edge set of a graph $G$, then any two
  adjacent edges of distinct $R$-classes span a (not necessarily unique)
  square with opposite edges in the same equivalence class.
\end{obs}

In the remainder of this section we collect several basic properties of
USP-relations. These results have originally be obtained for the
relation $\delta$ in the context of graph products and later were
generalized to the unique square property for applications to Cartesian
graph bundles. Here we show that the statements remain true for
USP-relations.

\begin{lemma}\label{lem:incidence} 
  Let $R$ be a USP-relation on the edge set of a connected graph $G$.
  Then each vertex of $G$ is incident to at least one edge of each 
  $R$-class.
\end{lemma}
\begin{proof}
  This was shown for the relation $\delta$ in \cite{Feigenbaum85:CartProd}
  and later for equivalence relations $Q$ with the unique square property
  in \cite{Zmazek:02}. Obviously, the result remains true when equivalence
  classes of $Q$ are united, i.e., for any USP-relation $R$ coarser
    than $Q$. The assertion follows immediately from the definition of
    USP-relations.
\end{proof}

Hence, if $G$ is connected and $R$ is a USP-relation, then
$N_{\varphi}(u) \neq \emptyset$ and $\No(u)\neq \emptyset$ for all $u\in
V(G)$ and all $\vp\in R$. Thus, neither $\G$ nor $\Go$ has isolated
vertices.

\begin{lemma}\label{lem:bijection}
  Let $R$ be a USP-relation on $E(G)$ and let $[u,v]\in \vp \sqsubseteq R$.
  Then $R$ induces a bijection between the $\psi$-edges incident to $u$ and
  $\psi$-edges incident to $v$ for every $\psi\in R$.  Furthermore, the
  vertices $u$ and $v$ have the same $\psi$-degree for every $\psi\in R$
  with $\psi\neq\vp$.
\end{lemma}
\begin{proof}
  Again, the result was first proved for $\delta$ in
  \cite{Feigenbaum85:CartProd} and then for equivalence relations with
  the unique square property in \cite{Zmazek:02}. Now suppose $R$ is an
  USP-relation, i.e., there is an equivalence relation $Q\subseteq R$
  such that $Q$ has the unique square property.  Then each equivalence
  class $\chi\sqsubseteq R$ is the union of some $Q$-equivalence classes,
  $\chi=\bigcup_{\psi\subseteq\chi}\psi$. The result of \cite{Zmazek:02}
  guarantees the existence of a bijection of the $\psi$-edges incident to
  $u$ and the $\psi$-edges incident to $v$ for $[u,v]\in \vp\neq\chi$ for
  all $\psi\subset \chi$. Since $\psi\cap\psi'=\emptyset$ for any two 
  distinct classes $\psi, \psi'\sqsubseteq Q$ we conclude that the
  disjoint union of these bijections over the $\psi\subset \chi$ is a 
  bijection between the $\chi$-edges incident to $u$ and the $\chi$-edges
  incident to $v$ for any $[u,v]\in \vp\neq\chi$. Clearly, this bijection 
  is again induced by $R$. It follows immediately that the $\chi$-degrees
  of $u$ and $v$ are also the same for all $\chi\neq\vp$.
\end{proof}

The following result was proved in \cite{Imrich:94} assuming the square
property. The proof uses only the existence but not the uniqueness of these
squares. Thus, by Observation~\ref{obs:squares}, the result remains
  true for USP-relations:
\begin{lemma}\label{lem:nonempty_intersect}
  Let $R$ be a USP-relation on $E(G)$ that contains only two equivalence
  classes $\varphi,\ \overline\varphi$.  Then
  \begin{equation*}
    |V(\G^x)\cap V(\Go^y)|\geq 1
  \end{equation*}
  for all $x,y\in V(G)$.
\end{lemma}

If $R$ is a convex USP-relation, i.e., if $R$ is a product relation,
then $|V(\G^x)\cap V(\Go^y)|= 1$ for all $x,y\in V$ and all $\varphi\eqcl
R$ \cite{Imrich:94}. In Theorem~\ref{thm:prodrel} below we will show
that the converse is also true.

We will need also the following technical results:
\begin{lemma} \label{lem:subgraph} 
  Let $R$ be a USP-relation on the edge set $E(G)$ of a connected graph
  $G$.  Let $Q\subseteq R$ be an equivalence relation on $E(G)$ with unique
  square property.  For $[v,w]\in\chi\sqsubseteq Q$ and $\vp\sqsubseteq R$
  with $\vp\cap\chi=\emptyset$ let $H$ be the subgraph of $G$ with vertex
  set $V(\G^v)\cup V(\G^w)$ that contains only $\vp$-edges and
  $\chi$-edges, respectively, that is,
  \begin{equation*}
    E(H)=\{[x,y]\in\vp\mid x,y\in
    V(\G^v)\cup V(\G^w)\}\cup\{[x,y]\in\chi\mid x,y\in V(\G^v)\cup
    V(\G^w)\}\,.
  \end{equation*}  
  Then $Q$ restricted to $H$ has the unique square property on $H$.
\end{lemma}
\begin{proof}
  It suffices to show that for any two adjacent edges $e=[x,y],f=[x,z]\in E(H)$
  with $(e,f)\notin Q$ the vertex $u$ of the unique square $(x,y,u,z)$
  with opposite edges in the same equivalence class
  spanned by $e$ and $f$ in $G$ is already contained in $V(H)$. 
  
  W.l.o.g.\ let $x\in V(\G^v)$ and $e\in\alpha\subseteq\varphi$, 
  $\alpha\sqsubseteq Q$. 
  Hence, $[u,z]\in\alpha\subseteq\vp$ and therefore 
  $u\in V(\G^z)\subseteq V(\G^v)\cup V(\G^w)=V(H)$.
\end{proof}

\section{Results}

\subsection{Equitable Partitions}

\begin{lemma} \label{lem:neighbors} Let $G$ be a graph and let
  $\vp\neq\psi$ be two equivalence classes of a USP-relation $R$ and let
  $v,w\in V(G)$. Then all vertices of $\G^v$ have the same number of
  incident $\psi$-edges connecting $\G^v$ and $\G^w$.  More formally,
  \begin{equation*} 
    |\Npsi(v)\cap V(\G^w)|=|\Npsi(x)\cap V(\G^w)|
  \end{equation*}
  holds for all $x\in V(\G^v)$. 
\end{lemma}

\begin{proof}
  First, we show that $\Npsi(v)\cap V(\G^w)=\emptyset$
  if and only if
  $\Npsi(x)\cap V(\G^w)=\emptyset$ holds for all $x\in V(\G^v)$.
 
  W.l.o.g., let $[v,w]\in\psi$ and consider an arbitrary vertex $x\in
  V(\G^v)$. Then there is a path $P:=(v=v_0,v_1,\ldots v_k=x)$ from $v$ to
  $x$ in $\G^v$. Recalling Observation~\ref{obs:squares}, we can construct
  a walk $Q=(w=w_0,w_1,\ldots,w_k)$ such that $[v_i,w_i]\in\psi$ for all
  $0\leq i\leq k$ and $[w_i,w_{i+1}]\in\vp$ for all $0\leq i\leq k-1$. Then
  $w_k\in N_{\psi}(x)$ and $w_k\in V(\G^w)$ and therefore, $\Npsi(x)\cap
  V(\G^w)\neq\emptyset$. Since $x\in V(\G^v)$ was arbitrarily chosen, we
  can conclude that $\Npsi(x)\cap V(\G^w)\neq\emptyset$ holds for all $x\in
  V(\G^v)$. Conversely, if $\Npsi(x)\cap V(\G^w)\neq\emptyset$ holds for
  all $x\in V(\G^v)$, this is trivially fulfilled also for $x=v$. Thus, we
  have $\Npsi(v)\cap V(\G^w)=\emptyset$ if and only if $\Npsi(x)\cap
  V(\G^w)=\emptyset$ holds for all $x\in V(\G^v)$.

  Now Suppose that $\Npsi(v)\cap V(\G^w)\neq\emptyset$. W.l.o.g., let
  $[v,w]\in\psi$. Since $R$ is a USP-relation, there is some relation
  $Q\subseteq R$ that has the unique square property, and $\psi\sqsubseteq
  R$ is the disjoint union of some equivalence classes $\chi\sqsubseteq Q$,
  $\psi=\bigcup_{\chi\subseteq\psi}\chi$. Thus, we have
  \begin{equation*}
    |\Npsi(x)\cap V(\G^w)|=\sum_{\chi\subseteq\psi}|N_\chi(x)\cap V(\G^w)|
  \end{equation*} 
  Therefore, it suffices to show that 
  $|N_\chi(v)\cap V(\G^w)|=|N_\chi(x)\cap V(\G^w)|$ holds for all 
  $x\in V(\G^v)$ and all $\chi\sqsubseteq Q$ with $\chi\subseteq\psi$.  
  In the following we denote with $N_{\mid H}$ the intersection of some 
  set $N\subseteq V(G)$ and the vertex set of a given subgraph 
  $H\subseteq G$.

  Suppose first, $x\in\N(v)$, i.e., $[v,x]\in\vp\neq\psi$.  By
  construction, $\vp\cap\chi=\emptyset$ holds for all $\chi\subseteq\psi$.
  Using the same arguments as before, we can conclude that $N_{\chi}(x)\cap
  V(G^w)=\emptyset$ if and only if $N_{\chi}(v)\cap V(G^w)=\emptyset$.
  Therefore, assume $N_{\chi}(v)\cap V(G^w)\neq\emptyset$.  Let $H$ be the
  subgraph of $G$ with vertex set $V(\G^v)\cup V(\G^w)$ defined as in
  Lemma~\ref{lem:subgraph}.  Then the restriction of $Q$ to $H$ satisfies
  the unique square property on $H$.
 
  If $\G^v=\G^w$, we can conclude by Lemma~\ref{lem:bijection} that
  \begin{equation} \label{eq:notind}
    |N_\chi(v)\cap V(\G^w)|=|N_\chi(v)_{\mid H}|=|N_\chi(x)_{\mid H}|=
    |N_\chi(x)\cap V(\G^w)|.
  \end{equation}
  
  Assume now $\G^v\neq\G^w$, and hence $V(\G^v)\cap V(\G^w)=\emptyset$.
  Thus $|N_{\chi}(y)_{\mid H}|=|N_\chi(y)\cap V(\G^v)|+|N_\chi(y)\cap V(\G^w)|$
  holds for all $y\in V(\G^v)$ and therefore, we can conclude again from 
  Lemma~\ref{lem:bijection} and Equation~\eqref{eq:notind}
  \begin{equation*}
    |N_\chi(x)\cap V(\G^w)|=|N_\chi(x)_{\mid H}|-|N_\chi(x)\cap V(\G^v)|
    =
    |N_\chi(v)_{\mid H}|-|N_\chi(v)\cap V(\G^v)|=|N_\chi(v)\cap V(\G^w)|\,,
  \end{equation*}
  which implies $|\Npsi(x)\cap V(\G^w)|=|\Npsi(v)\cap V(\G^w)|$.

  If $v$ and $x$ are connected by a path in $\G^v$, the assertion follows
  by induction on the length of the path.
\end{proof}

\begin{corollary}
  Let $G$ be a connected graph and $R$ be a USP-relation on $E(G)$.  Then
  $P^R_{\vpo}$ is an equitable partition of the graph $\G$ for every
  equivalence class $\vp$ of $R$.
\end{corollary}
\begin{proof} This follows immediately from $V(G)=V(\G)$, the fact that
  $P^R_{\vpo}$ defined in Eq.(\ref{eq:Povp}) is a partition of $V(G)$, and
  Lemma~\ref{lem:neighbors}.
\end{proof}

If $\G$ is an induced subgraph of $G$ we have $\G/\mathcal{P}^R_{\vpo}
\cong \mathcal{N}(G/\mathcal{P}^R_{\vpo})$ which follows from the fact
that $[\Go^x,\Go^y]$ is an edge in $\mc N(G/\mc P^R_{\vpo})$ if and only if
there is an edge in $G$ connecting a vertex in $V(\Go^x)$ with a vertex in
$V(\Go^y)$. This edge must be in $\vp$, since otherwise $\Go^x=\Go^y$, and
hence it is in $\G$.

\begin{remark}
  The quotient graphs $B_{\vpo}:=\G/\mc P_{\vpo}^R$ provide a direct
  connection to the theory of graph bundles since $B_{\vpo}$ coincides with
  the base graph of the bundle presentation $(G,p_{\vpo},B_{\vpo})$ of $G$
  provided $\vpo$ is $2$-convex \cite{Imrich:97,PZZ-01}. We recall that a
  subgraph $H\subseteq G$ is $2$-convex w.r.t. $G$ if all shortest $G$-paths of length
  $\leq 2$ connecting pairs of vertices in $H$ are contained in $H$.  An
  equivalence class $\vp \eqcl R$ is said to be $2$-convex if all connected
  components $\G^x$ of $\G$ are $2$-convex w.r.t. $G$. Moreover, it can easily be shown
  that $G$ has a graph bundle presentation $(G,p,\G/\mc P_{\vpo}^R)$ over a
  simple base if and only if $p:\G \rightarrow \G/\mc P_{\vpo}^R$ is a
  covering projection, i.e., a locally bijective homomorphism
  \cite{FPT:08,Imrich:97,PZZ-01}.
\end{remark}

\begin{theorem} \label{thm:equitpart} Let $R$ be a USP-relation on the edge
  set $E(G)$ of a connected graph $G$.  Then $\mc P^R$ defined in
  Eq.(\ref{eq:mcPR}) is an equitable partition of $G$.
\end{theorem}

To prove the Theorem, we first show the following:

\begin{lemma} \label{lem:neighborhood-cut} Let $G$ be a connected graph and
  $R$ be a USP-relation on $E(G)$.  Then for an arbitrary equivalence class
  $\vp$ of $R$ holds:
  \begin{itemize}
   \item[(1)] $\N(x)\cap \inters(y)\neq\emptyset$ if and only if 
     $\N(u)\cap \inters(y)\neq\emptyset$ for all $u\in \inters(x)$.
   \item[(2)] $\N(x)\cap \inters(y)\neq\emptyset$ implies 
     $\N(x)\cap \inters(y)=\N(x)\cap V(\Go^y)$.
  \end{itemize}
\end{lemma}

\begin{proof}
  \begin{itemize}
  \item[(1)]
  Let $\N(x)\cap \inters(y)\neq\emptyset$ and hence, $\N(x)\cap
  V(\Gopsi^y)\neq \emptyset$ for all $\psi \eqcl R$. Thus, there exists a
  vertex $z\in V(G)$ with $[x,z]\in\vp$ such that $z\in V(\Gopsi^y)$ for
  all $\psi \eqcl R$. Note, it holds $z\in V(\Gopsi^x)$ since for all
  $\vp\neq \psi$ there is a path that is not in $\psi$ which is the
  particular edge $[x,z]\in\vp$ . Therefore, $\Gopsi^x=\Gopsi^y$ for all
  $\psi\neq \vp$.
  
  Now let $u\in \inters(x)$. Hence $u\in V(\Gopsi^x)=V(\Gopsi^y)$ for all
  $\psi\neq \vp$. From Lemma~\ref{lem:neighbors} and the fact that
  $\N(x)\cap V(\Go^y)\neq \emptyset$, we can conclude that $\N(u)\cap
  V(\Go^y)\neq \emptyset$, i.e., there exists a vertex $w\in V(\Go^y)$
  such that $[u,w]\in\varphi$. This implies $w\in
  V(\Gopsi^u)=V(\Gopsi^y)$ for all $\psi\neq \vp$ and therefore $w\in
  \inters(y)$, hence $\N(u)\cap \inters(y)\neq\emptyset$. Conversely, if
  $\N(u)\cap \inters(y)\neq\emptyset$ for all $u\in \inters(x)$, this is
  trivially fulfilled for $u = x$.
  \item[(2)]
  Let $z\in\N(x)\cap \inters(y)$, that is,
  $z\in\N(x)$ and $z\in V(\Gopsi^y)$ for all $\psi \eqcl R$, in
  particular, $z\in V(\Go^y)$. Hence, $z\in\N(x)\cap V(\Go^y)$ and
  therefore we have $\N(x)\cap \inters(y)\subseteq \N(x)\cap V(\Go^y)$.
  Now, let $z\in \N(x)\cap V(\Go^y)$, which is equivalent to
  $[x,z]\in\varphi$ and $z\in V(\Go^y)$. It follows $z\in V(\Gopsi^x)$
  for all $\psi\neq \vp$ and thus $z\in V(\Gopsi^y)$ for all $\psi \eqcl
  R$ since $\N(x)\cap \inters(y)\neq\emptyset$. Hence, $z\in \N(x)\cap
  \inters(y)$ and therefore $\N(x)\cap V(\Go^y)\subseteq\N(x)\cap
  \inters(y)$, from which we can conclude equality of the sets.
  \end{itemize}
\end{proof}

\begin{proof}[Proof of Theorem~\ref{thm:equitpart}]
  By construction $\mc P^R$ is a partition of $V(G)$.
  It remains to show that this partition is equitable, that is, we have to
  show that for arbitrary $u,x,y\in V(G)$ with $u\in \inters(x)$ holds
  \begin{equation}	\label{eq:equitable}
    |N_G(u)\cap \inters(y)|=|N_G(x)\cap \inters(y)|.
  \end{equation}
  Notice, that for arbitrary $x\in V(G)$ holds $N_G(x)=\bigcup_{\vp \eqcl
    R} \N(x)$ and $\N(x)\cap \Npsi(x)=\emptyset$ for $\vp\neq \psi$. Hence we
  have $|N_G(x)\cap \inters(y)|=\sum_{\vp \eqcl R} |\N(x)\cap \inters(y)|$.
  Therefore, it suffices to show
  \begin{equation*}
    |\N(u)\cap \inters(y)|=|\N(x)\cap \inters(y)|\quad\forall \vp\eqcl R
  \end{equation*}
  to prove Eq.~\eqref{eq:equitable}. This equality, however, follows
  immediately from Lemma~\ref{lem:neighborhood-cut} together with
  Lemma~\ref{lem:neighbors}.
\end{proof}

\subsection{Product Structure of Quotient Graphs}

Product structures and equitable partitions are compatible in the following
sense: 
\begin{proposition} \cite{BFF+12}
\label{prop:BFF}
  Let $G=\BOX_{i=1}^n G_i$ and let $\pi_i$ be an equitable partition on
  $G_i$. Then there is an equitable partition $\pi$ of $G$ such that 
  \begin{equation*}
    \BOX_{i=1}^n \overrightarrow{(G_i/\pi_i)} = \overrightarrow{(G/\pi)}
  \end{equation*}
\end{proposition}

Since the Cartesian product of the underlying undirected and unweighted graphs is obtained
by simply omitting the weights, we also have 
\begin{equation*}
  \BOX_{i=1}^n (G_i/\pi_i) = (G/\pi)
\end{equation*}
for the same equitable partition $\pi$ of $G$. 

Our next result shows that equitable partitions constructed above arrange 
themselves as a special case of Prop.~\ref{prop:BFF}. 

\begin{theorem} \label{thm:prod}
  Let $R$ be a USP-relation on the edge set $E(G)$ of a connected graph $G$.
  Then
  \begin{equation*}
  G/\mathcal P^R\cong\BOX_{\vp \eqcl R} \G/\mathcal P^R_{\vpo}.
  \end{equation*}
\end{theorem}
\begin{proof}
  Let $\varphi_1,\ldots \varphi_n$ denote the equivalence classes of $R$.
  Let $x,v_1,\ldots,v_n\in V(G)$, where the $v_i$ need not necessarily be
  distinct. If $x\in V(\Goi^{v_i})$ for all $i=1,\ldots n$ then
  $\inters(x)=\bigcap_{i=1}^n V(\Goi^{v_i})$.

  Remark, that for $1\leq i\leq n$ the vertex set of $\Gi/\mc P_{\vpo_i}^R$
  is given by $V(\Gi/\mc P_{\vpo_i}^R)=\{\Goi^{v_i}\mid v_i\in V(G)\}$.
  Hence, we have \[V(\BOX_{i=1}^n \Gi/\mathcal
  P_{\vpo_i}^R)=\left\{(\Goe^{v_1},\ldots,\Gon^{v_n})\mid v_i\in V(G),
    i=1,\ldots,n\right\},\] where
  $(\Goe^{v_1},\ldots,\Gon^{v_n})=(\Goe^{u_1},\ldots,\Gon^{u_n})$ if and
  only if $u_i\in V(\Goi^{v_i})$ for all $i=1,\ldots,n$.
  
  We define a mapping $V(G/\mathcal P^R) \rightarrow
  V(\BOX_{i=1}^n \Gi/\mathcal P_{\vpo_i}^R)$ as follows: \[\inters(x)\mapsto
  (\Goe^{v_1},\ldots,\Gon^{v_n})\] iff $x\in V(\Goi^{v_i})$ for all
  $i=1,\ldots n$.
    
  For all $x\in V(G)$ there exist $v_i$, $i=1,\ldots,n$ such that $x\in
  V(\Goi^{v_i})$, e.g. choose $v_i=x$. And since from $x\in V(\Goi^{v_i})$
  and $x\in V(\Goi^{u_i})$ follows $\Goi^{v_i}=\Goi^{u_i}$, this mapping is
  well defined.
  
  Due to the fact that $x\in V(\Goi^{v_i})$ and $y\in V(\Goi^{v_i})$ implies
  $\Goi^x=\Goi^y$, we can conclude that this mapping is injective. To prove
  surjectivity, it suffices to show, that $\cap_{i=1}^n V(\Goi^{v_i})\neq
  \emptyset$ for arbitrary $v_i\in V(G)$. We show by induction for all
  $k\leq n$ holds $\cap_{i=1} ^k V(\Goi^{v_i}) \neq \emptyset$. For $k=1$ this
  is trivially fulfilled. Let $k\geq 1$ and suppose $\cap_{i=1} ^k
  V(\Goi^{v_i}) \neq \emptyset$. We have to show, that this implies 
  $\cap_{i=1}^{k+1} V(\Goi^{v_i}) \neq \emptyset$. 
  From the induction hypothesis, we can
  conclude there must be a vertex $x\in V(G)$ such that $x\in V(\Goi^{v_i})$
  for all $i=1,\ldots,k$ and hence $\cap_{i=1} ^k V(\Goi^{v_i})=\cap_{i=1} ^k
  V(\Goi^x)$ for all $i=1,\ldots,k$. Therefore, we have to show 
  \begin{equation}	\label{eq:subset}
    V(\Gkp^x)\subseteq\bigcap_{i=1}^k V(\Goi^x).
  \end{equation}
  From that and Lemma~\ref{lem:nonempty_intersect} we can conclude
  $\emptyset\neq V(\Gkp^x)\cap V(\Gokp^{v_{k+1}})\subseteq
  \bigcap_{i=1}^{k+1} V(\Goi^{v_i})$, from what the assumption follows.
  
  Let $y\in V(\Gkp^x)$. Then there exists a path $Q$ from $x$ to $y$ such
  that all edges of $Q$ are in class $\varphi_{k+1}$. Clearly, they are not
  in class $\varphi_i$ for $i=1,\ldots k$ and therefore $y\in V(\Goi^x)$
  for all $i=1,\ldots k$, from what Equation~\eqref{eq:subset} and finally
  surjectivity follows.

  \begin{figure}[tbp]
    \centering
    \includegraphics[bb= 112 220 570 461, scale=0.75]{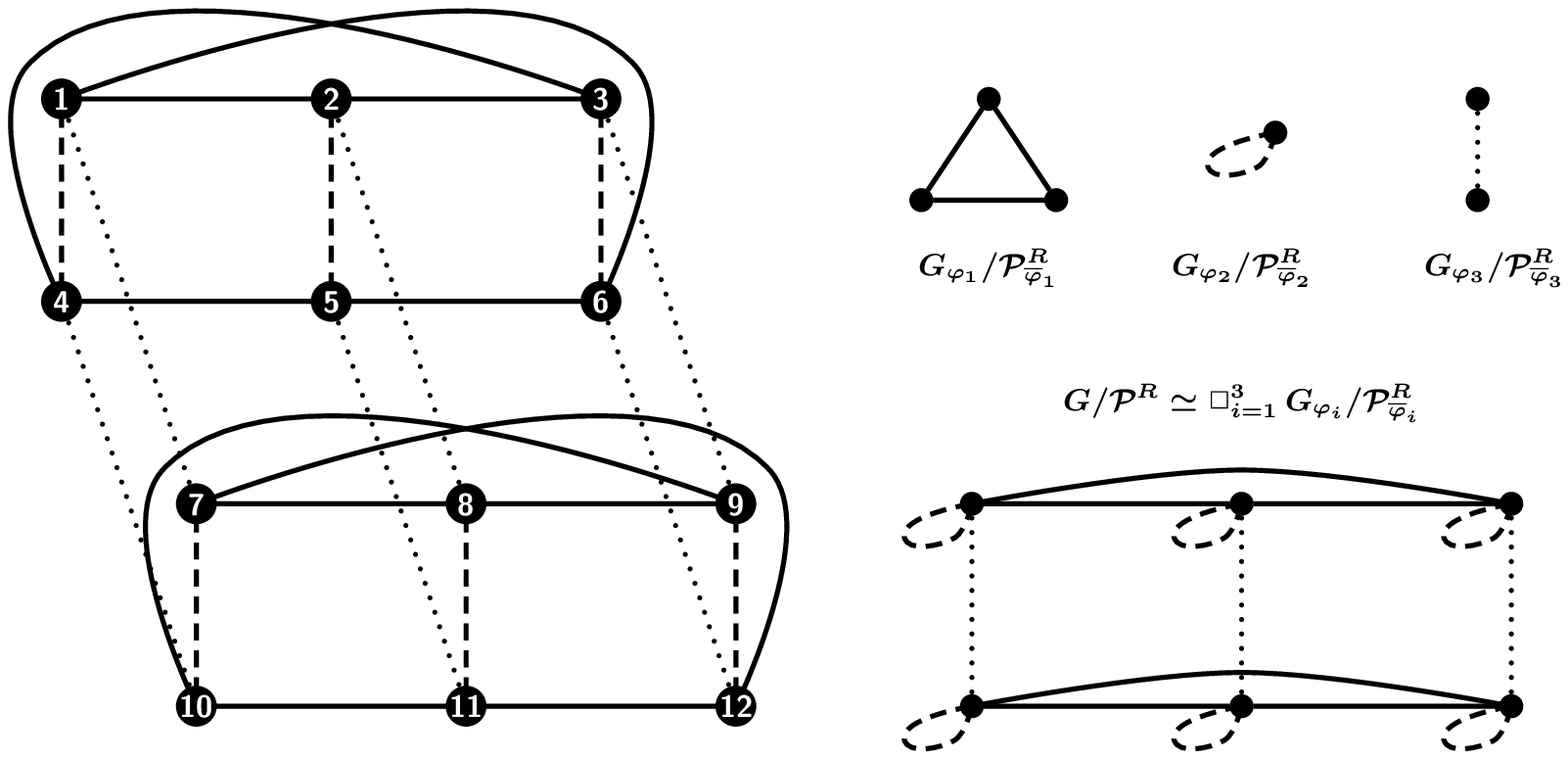}
    \caption{The equivalence relation $R$ on $E(G)$ with 
      equivalence classes $\vp_1$ (solid), $\vp_2$ (dashed), and 
      $\vp_3$ (dotted) has the unique square property. We have
      $\mc P^R_{\vpo_1}=\{\{1,4,7,10\},\{2,5,8,11\},\{3,6,9,12\}\}$,
      $\mc P^R_{\vpo_2}=\{V(G)\}$,
      $\mc P^R_{\vpo_3}=\{\{1,2,\ldots,6\},\{7,8,\ldots,12\}\}$
      and  $\mc P^R=\{\{1,4\},\{2,5\},\{3,6\},\{7,10\},\{8,11\},\{9,12\}\}$. 
      The corresponding quotient graphs 
      $G_{\vp_i}/\mc{P}_{\vpo_i}^R$, $i=1,2,3$ and 
      the product graph $G/\mc P^R$ are shown on the right-hand side.
    }          
    \label{fig:exmplProd}
  \end{figure}
  
  It remains to verify the isomorphism property, that is
  $[\inters(x),\inters(y)]$ is an edge in $G/\mc P^R$ if and only if
  $[(\Goe^x,\ldots,\Gon^x),(\Goe^y,\ldots,\Gon^y)]$ is an edge in
  $\BOX_{i=1}^n\Gi/\mc P_{\vpo_i}^R$. Let $[\inters(x),\inters(y)]\in
  E(G/\mc P^R)$, that is, there exists a vertex $x'\in\inters(x)$ and a
  vertex $y'\in\inters(y)$ such that $[x',y']$ is an edge in $G$ and
  therefore in $\varphi_i$ for some $i$, $1\leq i\leq n$. This implies
  $\Goj^{x}=\Goj^{x'}=\Goj^{y'}=\Goj^y$ for all $j\neq i$ and
  $[\Goi^x,\Goi^y]\in E(\Gi/\mc P_{\vpo_i}^R)$. Thus,
  $[(\Goe^x,\ldots,\Gon^x),(\Goe^y,\ldots,\Gon^y)]$ is an edge in
  $\BOX_{i=1}^n\Gi/\mc P_{\vpo_i}^R$. Conversely, let
  $[(\Goe^x,\ldots,\Gon^x),(\Goe^y,\ldots,\Gon^y)]\in E(\BOX_{i=1}^n\Gi/\mc
  P_{\vpo_i}^R)$. There must be an $i$, $1\leq i\leq n$ such that
  $[\Goi^x,\Goi^y]\in E(\Gi/\mc P_{\vpo_i}^R)$ and $\Goj^x=\Goj^y$ for all
  $j\neq i$. $[\Goi^x,\Goi^y]\in E(\Gi/\mc P_{\vpo_i}^R)$ implies that there
  exists a vertex $x'\in V(\Goi^x)$ and a vertex $y'\in V(\Goj^y)$ such
  that $[x',y']\in\varphi_i$ in $G$. From Lemma~\ref{lem:neighbors}, we can
  conclude that there exists a vertex $z\in V(\Goi^y)$ such that
  $[x,z]\in\varphi_i$. This in turn implies $z\in V(\Goj^x)=V(\Goj^y)$ for
  all $j\neq i$ and thus, $z\in\inters (y)$.
  Hence, $[\inters(x),\inters(y)]$ is an edge in $G/\mc P^R$.  
\end{proof}

\begin{figure}[t]
  \centering
  \includegraphics[bb=58 266 537 580, scale=0.6]{./b2.ps}
  \caption{The left panel shows a graph with a USP-relation $R$ whose 
    equivalence classes are highlighted by dashed and solid edges. 
    The corresponding quotient graph $G/\mc P^R$ and its Cartesian prime
    factors are shown on the right-hand side.}
 \label{fig:OtherExample}
\end{figure}

\begin{corollary}
  Suppose the conditions of Theorem~\ref{thm:prod} are satisfied.  If
  furthermore $\Go$ is an induced subgraph of $G$ for all $\vp \eqcl R$
  then
  \begin{equation*}
    G/\mathcal P^R\cong\BOX_{\vp\eqcl R} \mc N(G/\mathcal P_{\vpo}^R).
  \end{equation*}
\end{corollary}

\begin{proof}
  It suffices to show that $G/\mc P^R$ has no loops if all $\Go$ are
  induced. We will prove this by contradiction. Therefore, assume that $G/\mc
  P^R$ contains a loop $[\inters (x),\inters (x)]$ for some $x\in V(G)$.
  Hence, there are vertices $y,z\in \inters(x)$ with $[y,z]\in E(G)$.
  Clearly, $[y,z]\in\varphi$ for some $\vp \eqcl R$. But since $y,z\in
  V(\Go^x)$ it follows that $\Go$ is not induced, a contradiction.
\end{proof}

\begin{corollary}
  If the conditions of Theorem~\ref{thm:prod} are satisfied, then
  \begin{equation*}
    \overrightarrow{G/\mc P^R}\cong\BOX_{\vp\eqcl R} 
    \overrightarrow{\G/\mc P_{\vpo}^R}.
  \end{equation*}
\end{corollary}

\begin{proof}
  Since the underlying undirected and unweighted graphs of
  $\overrightarrow{G/\mathcal{P}^R}$ and
  $\overrightarrow{\G/\mathcal{P}_{\vpo}^R}$ are exactly $G/\mathcal{P}^R$
  and $\G/\mathcal{P}_{\vpo}^R$, respectively, it suffices to show that the
  weights are transferred as in Eq.(\ref{eq:weights}).  This follows
  immediately from Lemma~\ref{lem:neighborhood-cut} and
  Lemma~\ref{lem:neighbors} and the fact that $|N_G(x)\cap
  V_R(y)|=\sum_{\vp\eqcl R}|\N(x)\cap V_R(y)|$.
\end{proof}

With the help of the results obtained in this section we can strengthen 
a useful result of \cite{Imrich:94}:

\begin{theorem} \label{thm:prodrel} Let $Q$ be a USP-relation on the edge
  set $E(G)$ of a connected graph $G$ and let $\vp \eqcl Q$. Then 
  $|V(\G^x)\cap V(\Go^y)|=1$ holds for all $x,y\in V(G)$ if and only if 
  $R = \{\vp, \vpo\}$ is a product relation.
\end{theorem}

\begin{proof}
	It has been shown in \cite{Imrich:94} that for product relations, 
	that  i.e., convex USP-relations, holds $|V(\G^x)\cap V(\Go^y)|=1$ for all
	$x,y\in V(G)$. It remains to show, therefore, that converse is also 
	true.
  
  Notice that $\overline{\vpo}=\vp$. Hence the equitable partition induced
  by $R$ is $\mathcal{P}^R=\{V(\G^x)\cap V(\Go^x)\mid x\in V(G)\}$. 
  By assumption, $\mathcal{P}^R$ consists exclusively of singletons. Thus
  $G=G/\mc P^R$.  Recall that $\mc P^R_{\vpo}=\{V(\Go^x)\mid x\in V(G)\}$
  and $\mc P^R_{\vp}=\{V(\G^x)\mid x\in V(G)\}$ are the equitable
  partitions of the graphs $\G$ and $\Go$ respectively.  For arbitrary
  $y\in V(G)$ let $\mc P_{\vpo}^R(y)$ denote the restriction of $\mc
  P_{\vpo}$ to the connected component $\G^y$ of $\G$, that is $\mc
  P_{\vpo}^R(y)=\{V(\Go^x)\cap V(\G^y)\mid x\in V(G)\}$.  From
  Lemma~\ref{lem:nonempty_intersect}, Lemma~\ref{lem:neighbors} and the
  definition of the quotient graphs, we can conclude that the mapping
  $\Go^x\cap \G^y\mapsto\Go^x$ defines an isomorphism $\G/\mc
  P_{\vpo}^R\cong \G^y/\mc P_{\vpo}^R(y)$ for all $y\in V(G)$ and since
  $|V(\G^x)\cap V(\Go^y)|=1$ holds for all $x,y\in V(G)$, we even have
  $\G/\mc P_{\vpo}^R\cong \G^y$. Analogously, it follows $\Go/\mc
  P_{\vpo}^R\cong \Go^y$ for all $y\in V(G)$. Thus, $ G\cong
  \G^x\BOX\Go^y$ for all $x,y\in V(G)$, demonstrating that 
   $R = \{\vp, \vpo\}$ is a product relation.
  \end{proof}

\subsection{Refinements and Coarse Graining}

Given a graph $G$ and a nontrivial USP-relation $R$ on $E(G)$ it will
often be the case that $G/\mc P^R$ has no ``real'' product structure, as
in the example of Figure~\ref{fig:CounterUnion}. Here, $V(\Goi)=V(G)$ for
each of the four equivalence classes, so that $G/\mc P^R$ is the trivial
graph $\mc LK_1$ consisting of a single vertex with a loop. In
Section~\ref{sec:prelim}, we have shown that a coarse graining $S$ of a
USP-relation $R$ in general leads to a refinement $\mc P^S$ of the vertex
partition $\mc P^R$. Hence we can expect to obtain larger quotient graph
$G/\mc P^S$ with a ``richer'' product structure. This is indeed sometimes
the case as shown by the example in Fig.~\ref{fig:join}.

\begin{figure}[t]
  \centering
  \includegraphics[bb= 30 590 530 740, scale=0.75]{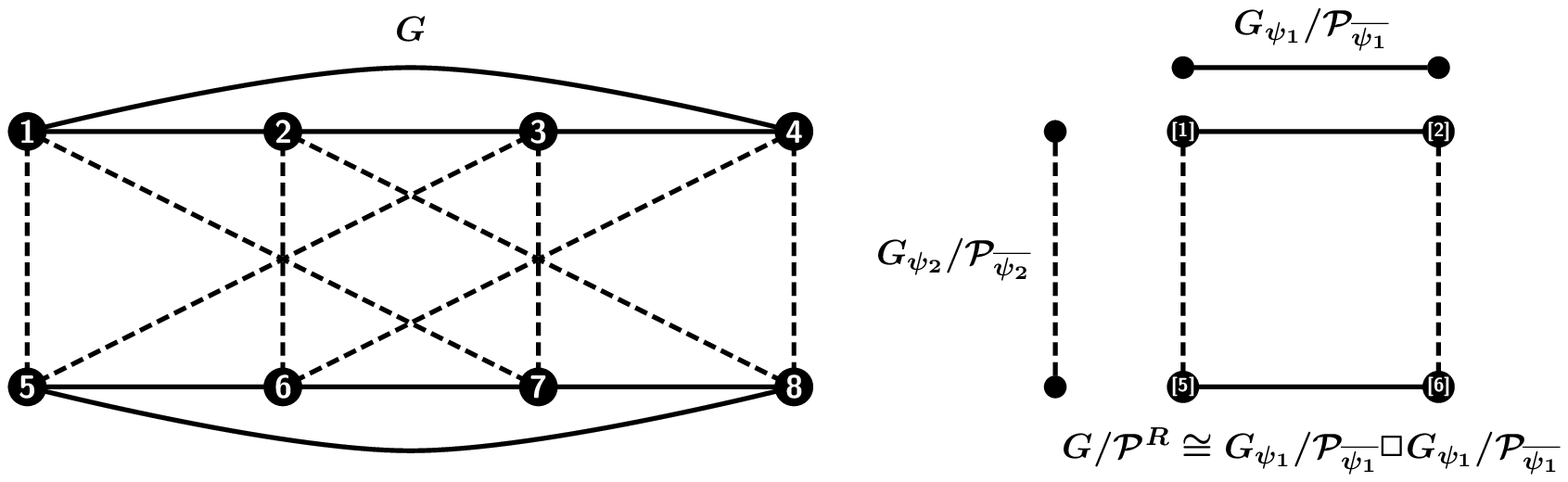}
  \caption{The coarse graining $R=\{\psi_1=\vp_1\cup
    \vp_2,\psi_2=\vp_3\cup\vp_4\}$ obtained from the equivalence relation 
    $Q$ of Fig.~\ref{fig:CounterUnion} generates the quotient graph 
    $G/\mc P^R\cong K_2\Box K_2$ with non-trivial product structure.
  }
  \label{fig:join}
\end{figure}

However, a coarser relation $S\supseteq R$ does not always lead to a
partition $\mc P^S$ that is strictly finer than $\mc P^R$, see
Fig.~\ref{fig:exmpljoin} for an example. In this section we therefore
explore the conditions under which a \emph{strictly} finer partition $\mc
P^S$ of the vertex set is obtained by a coarser equivalence relation
$S\supseteq R$.

\begin{figure}[t]
  \centering
  \includegraphics[bb= 112 180 570 440, scale=0.75]{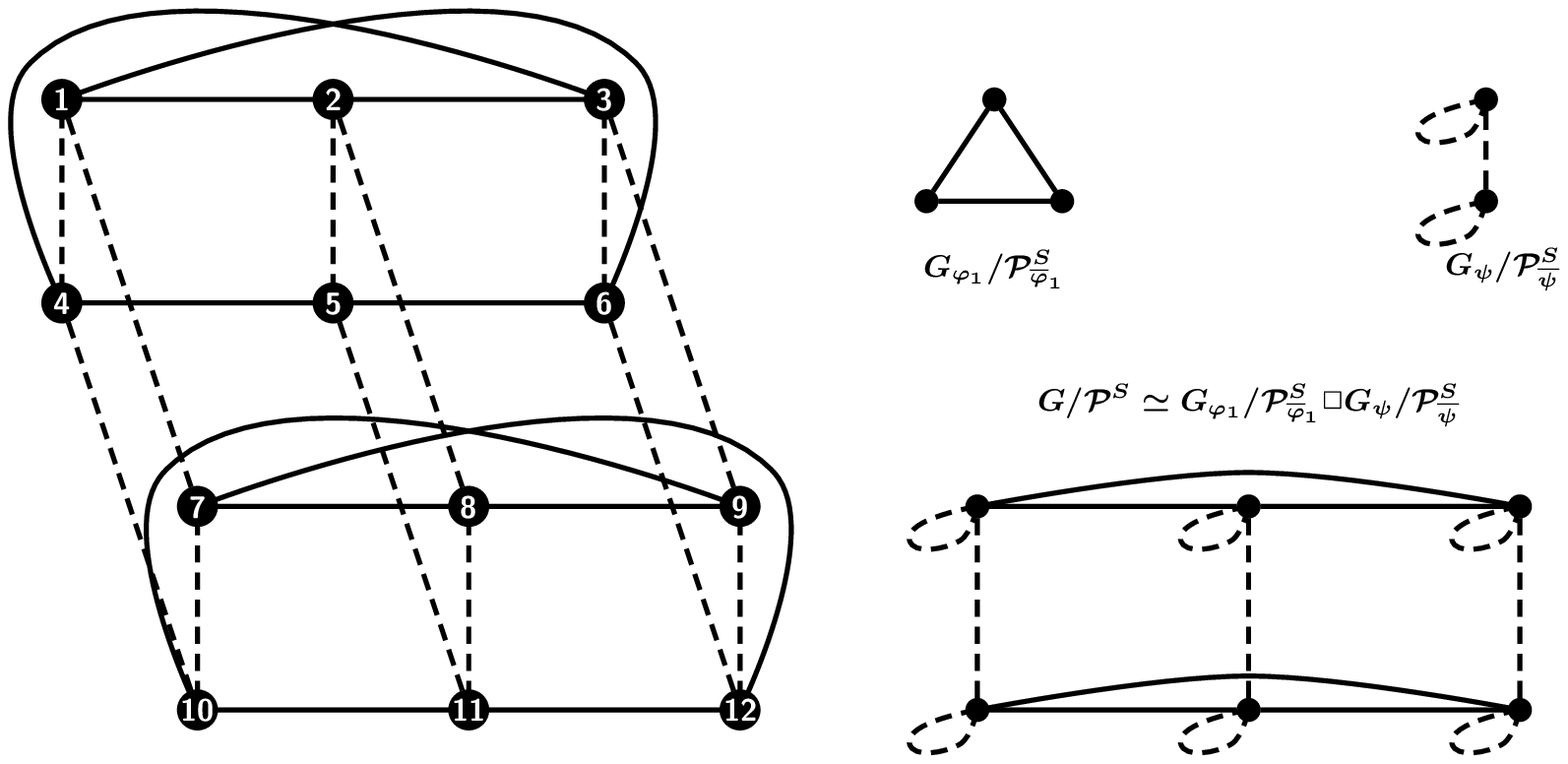}
  \caption{The coarse graining $S=\{\vp_1,\psi=\vp_2\cup\vp_3\}$ of
    the relation $R$ of Fig.~\ref{fig:exmplProd} leads to the same
    partition $\mc P^S=\mc P^R$ of $V(G)$ and thus to identical quotient
    graphs.}
  \label{fig:exmpljoin}
\end{figure}
  
\begin{proposition} \label{prop:components} Let $\varphi$ and $\psi$ be 
  two equivalence classes of a USP-relation $R$ on the edge set $E(G)$ of 
  a connected graph $G$.  Then for all $x\in V(G)$ holds
  \begin{equation*}
    V(G_{\varphi\cup\psi}^x)=
    \bigcup_{y\in V(G_{\varphi}^x)} V(G_{\psi}^y)=
    \bigcup_{y\in V(G_{\psi}^x)} V(G_{\varphi}^y).
  \end{equation*}
\end{proposition}

\begin{proof}
  It suffices to show the first equation.  Therefore, let
  $z\in\bigcup_{y\in V(G_{\varphi}^x)} V(G_{\psi}^y)$, that is, there
  exists a vertex $y'\in V(G_{\varphi}^x)$ such that $z\in
  V(G_{\psi}^{y'})$.  Hence, there is a path $P_{x,y'}$ from $x$ to
  $y'$ in $\varphi$ and a path $P_{y',z}$ from $y'$ to $z$ in $\psi$.
  Thus, $P_{x,y'}\cup P_{y',z}$ is a path from $x$ to $z$ in
  $\varphi\cup\psi$ and therefore $z\in V(G_{\varphi\cup\psi}^x)$ from
  which we can conclude $\bigcup_{y\in V(G_{\varphi}^x)} V(G_{\psi}^y)
  \subseteq V(G_{\varphi\cup\psi}^x)$.

  Now, let $z\in V(G_{\varphi\cup\psi}^x)$.  Clearly, the restriction of
  $R$ to $G_{\varphi\cup\psi}^x$ is an equivalence relation on
  $E(G_{\varphi\cup\psi}^x)$ with only two equivalence classes $\varphi$
  and $\psi$.  Therefore, by Lemma~\ref{lem:nonempty_intersect} we can
  conclude that $V(G_{\psi}^z)\cap V(G_{\varphi}^x)\neq\emptyset$.  Let
  $y\in V(G_{\psi}^z)\cap V(G_{\varphi}^x)$.  It follows that
  $G_{\psi}^z=G_{\psi}^y$ and thus, $z\in \bigcup_{y\in V(G_{\varphi}^x)}
  V(G_{\psi}^y)$ since in particular $y\in V(G_{\varphi}^x)$. From
  $V(G_{\varphi\cup\psi}^x)\subseteq \bigcup_{y\in V(G_{\varphi}^x)}
  V(G_{\psi}^y)$ we conclude equality of the sets.
\end{proof}

\begin{proposition}
  Let $R$ be a USP-relation on the edge set $E(G)$ of a connected graph $G$
  and let $\vp, \psi\sqsubseteq R$, $\vp\neq\psi$.  Then 
  $V(G_{\vpo}^x)\cap V(G_{\overline\psi}^x) = 
   V(G_{\overline{\varphi\cup\psi}}^x)$ 
  if and only if $V(G_{\varphi}^x)\cap V(\Go^x)\subseteq
  V(G_{\overline{\varphi\cup\psi}}^x)$.
\end{proposition}

\begin{proof}
  From Prop.~\ref{prop:components} we can compute
  $V(G_{\vpo}^x)\cap V(G_{\overline\psi}^x)
    = V(G_{\vpo}^x)\cap V(G_{\varphi\cup(\overline{\varphi\cup\psi})}^x)
    = \\V(\Go^x)\cap
    \left(\bigcup_{w\in V(G_{\varphi}^x)}
          V(G_{\overline{\varphi\cup\psi}}^w)\right)
    =\bigcup_{w\in V(G_{\varphi}^x)}
         \left(V(\Go^x)\cap V(G_{\overline{\varphi\cup\psi}}^w)\right)
  $.\\[0.2cm]
Notice that $V(G_{\overline{\vp\cup\psi}}^v)\subseteq V(\Go^x)$ if and only
if $v\in V(\Go^x)$, otherwise we would have 
$V(\Go^x)\cap V(G_{\overline{\varphi\cup\psi}}^v = \emptyset$. 
Therefore,  \\
$V(G_{\vpo}^x)\cap V(G_{\overline\psi}^x)
    =\bigcup_{w\in V(G_{\varphi}^x)\cap V(\Go^x)} 
    V(G_{\overline{\varphi\cup\psi}}^w)
    = V(G_{\overline{\varphi\cup\psi}}^x)\overset{\dotfill}{\cup}
    \left(\bigcup_{w\in \mc X} 
    V(G_{\overline{\varphi\cup\psi}}^w)\right)
$ 
with $\mc X = V(G_{\varphi}^x)\cap V(\Go^x)\setminus 
V(G_{\overline{\varphi\cup\psi}}^x)$\\[0.2cm]
Hence, we have $V(G_{\vpo}^x)\cap V(G_{\overline\psi}^x)=
V(G_{\overline{\varphi\cup\psi}}^x)$ if and only if $\mc X = \emptyset$
which is equivalent to $V(G_{\varphi}^x)\cap V(\Go^x)\subseteq
V(G_{\overline{\varphi\cup\psi}}^x)$.
\end{proof}

\begin{proposition} \label{lem:subsets}
  Let $R$ be a USP-relation on the edge set $E(G)$ of a connected graph $G$ 
  with two distinct equivalence classes $\varphi,\psi\eqcl R$.
  \begin{itemize}
  \item[(1)] If $V(\G^x)\subseteq V(G_{\psi}^x)$ for some $x\in V(G)$
    then $V(\G^y)\subseteq V(G_{\psi}^x)$
    holds for all $y\in V(G_{\psi}^x)$.
  \item[(2)] If $V(\G^x)\subseteq V(\Go^x)$ for some $x\in V(G)$ 
    then $V(\Go^x)=V(G)$.
  \item[(3)] If $V(\Go^x)=V(G)$, $x\in V(G)$,
    then for all $y\in V(G)$ holds 
    $V(\G^y)\cap V(\Go^y)\subseteq V(G_{\overline{\varphi\cup\psi}}^y)$
    if and only if $V(\G^y)\subseteq V(G_{\overline{\varphi\cup\psi}}^y)$.
  \end{itemize}
\end{proposition}

\begin{proof}
  \begin{itemize}
  \item[(1)] Let 
    $X:=\left\{v\in V(G_{\psi}^x)\mid V(\G^v)\subseteq V(G_{\psi}^x)\right\}$.
    If $V(\G^x)\subseteq V(G_{\psi}^x)$ then $X\neq\emptyset$.
    Suppose $V(G_{\psi}^x)\setminus X\neq\emptyset$.
    By connectedness of $G_{\psi}^x$, there exists some vertices 
    $y\in V(G_{\psi}^x)\setminus X$
    and $v\in X$ such that $[v,y]\in\psi$. Clearly, $y\notin V(\G^v)$.
    Since $y\notin X$, there exists a vertex 
    $w\in V(\G^y)\setminus V(G_{\psi}^x)$.
    Since $[v,y]\in\psi$, we can use Lemma~\ref{lem:neighbors} to conclude 
    that there exist a vertex $z\in V(\G^v)$ such that $[z,w]\in\psi$.
    This implies $w\in V(G_{\psi}^z)=V(G_{\psi}^x)$, since $v\in X$, a 
    contradiction.
  \item[(2)] Let $V(\G^x)\subseteq V(\Go^x)$ and suppose 
    $V(G)\setminus V(\Go^x)\neq\emptyset$.
    By connectedness of $G$, there exist vertices $v\in V(\Go^x)$ and 
    $y\in V(G)\setminus V(\Go^x)$
    such that $[v,y]\in E(G)$. Obviously, $[v,y]$ must be in $\varphi$.
    Hence, $y\in V(\G^w)$. From the first assertion, we conclude that 
    this implies $y\in V(\Go^x)$, a contradiction.
  \item[(3)] Clear.
  \end{itemize}
\end{proof} 

 We conclude our presentation by summarizing conditions under which the
  joining of two equivalence classes of a USP-relation does not affect the
  partitioning of the vertex set.

\begin{corollary}
  Let $R$ be a USP-relation on the edge set $E(G)$ of a connected graph $G$
  with two distinct equivalence classes $\varphi,\psi\eqcl R$ and denote by
  $S$ be the USP-relation obtained from $R$ by joining $\varphi$ and
  $\psi$. Then:
  \begin{itemize}
  \item[(1)] $\mc P^R=\mc P^S$ if $\vp$ or $\psi$ belong to a factor of $G$.
  \item[(2)] If here is a vertex $x\in V(G)$ with 
    $V(\G^x)\subseteq V(\Go^x)$ then 
    $\mc P^R=\mc P^S$ if and only if 
    $V(\G^y)\subseteq V(G_{\overline{\varphi\cup\psi}}^y)$
    holds for all $y\in V(G)$.
  \item[(3)] If $\mc P^R=\mc P^S$ then $\vp\cup\psi$ belongs to a factor 
    of $G$ if and only if
    both $\vp$ and $\psi$ belong to a factor of $G$.
  \end{itemize}
\end{corollary}

\begin{proof}
 \begin{itemize}
 \item[(1)]   W.l.o.g., let $\varphi$ correspond to a factor of $G$. Then it 
    holds $V(\G^x)\cap V(\Go^x)=\{x\}\subseteq
    V(G_{\overline{\vp\cup\psi}})$ 
    for all $x\in V(G)$, which implies the assertion.
 \item[(2)]   Follows immediately from Proposition~\ref{lem:subsets}.
  \item[(3)]  If both $\vp$ and $\psi$ correspond to factors, then clearly 
    $\vp\cup\psi$ also corresponds to a factor.
    Conversely, suppose $\vp\cup\psi$ correspond to a factor and suppose 
    $\mc P^R=\mc P^S$. Then 
    $|V(G_{\varphi\cup\psi}^x)\cap V(G_{\overline{\varphi\cup\psi}}^x)|=1$
    and $V(\G^x)\cap V(\Go^x)\subseteq V(G_{\overline{\vp\cup\psi}})$ holds
    for all $x \in V(G)$. Note that $V(\G^x)\subseteq V(G_{\vp\cup\psi}^x)$
    and hence  $V(\G^x)\cap V(\Go^x)=(V(\G^x)\cap V(\Go^x))\cap 
    V(G_{\vp\cup\psi}^x)\subseteq V(G_{\overline{\vp\cup\psi}})
    \cap V(G_{\vp\cup\psi}^x)$ holds for all $x\in V(G)$. This implies
    $V(\G^x)\cap V(\Go^x)=\{x\}$ for all $x\in V(G)$.
    By Theorem~\ref{thm:prodrel} we can now conclude that $\vp$ belongs 
    to a factor of $G$. 
    Analogously, it follows that $\psi$ belongs to a factor of $G$.  
 \end{itemize}
\end{proof}

\section*{Acknowledgments}
  This work was supported in part by the \emph{Deutsche
    Forschungsgemeinschaft} within the EUROCORES Programme EUROGIGA
  (project GReGAS) of the European Science Foundation.



\end{article}

\end{document}